\documentclass{article}
\usepackage[utf8]{inputenc}
\usepackage{authblk}

\title{Smoothed Analysis of the \AGP/}
\author[1]{Michael Gene Dobbins}
\author[2]{Andreas Holmsen}
\author[3]{Tillmann Miltzow}
\affil[1]{Department of Mathematical Sciences, Binghamton University}
\affil[2]{Department of Mathematical Sciences, KAIST }
\affil[3]{Department of Information and Computing Sciences, Utrecht University}
\date{}

\usepackage{graphicx}
\usepackage{subcaption}
\usepackage{xcolor}
\usepackage{amsmath,amsthm, amssymb,amsfonts}
\usepackage{fullpage}
\usepackage{url}
\usepackage{enumitem}
\usepackage{changepage}
\usepackage[numbers,sort&compress]{natbib}

\usepackage{thmtools} 
\usepackage{thm-restate}

\usepackage{comment}

\usepackage{ellipsis}

\newtheorem{thm}{Theorem}
\newtheorem*{thmNoNumber}{Theorem}
\newtheorem{lemma}[thm]{Lemma}

\newtheorem{question}{Question}

\def \ER/{\ensuremath{\exists \mathbb{R}}}
\def \NP/{\ensuremath{\mbox{NP}}}
\def \PSPACE/{\ensuremath{\mbox{PSPACE}}}

\def \MinkInfl/{Minkowski-Inflation}
\def \EdgeInfl/{Edge-Inflation}
\def \EdgeDefl/{Edge-Deflation}
\def \EdgePert/{Edge-Perturbation}
\def \VertexPert/{Vertex-Perturbation}
\def \VertexInfl/{Vertex-Inflation}
\def \NaiveAlgo/{Naive Algorithm}
\def \SmoothAna/{Smoothed Analysis}
\def \AGP/{Art Gallery Problem}
\def \realRAM/{real RAM}
\def \wordRAM/{word RAM}
\def \TuringMachine/{Turing Machine}

\newcommand {\R}{\ensuremath{\mathbb{R}}}
\newcommand {\Q}{\ensuremath{\mathbb{Q}}}

\newcommand {\E}{\ensuremath{\mathbb{E}}}
\newcommand {\N}{\ensuremath{\mathbb{N}}}
\newcommand {\Z}{\ensuremath{\mathbb{Z}}}

\newcommand {\disk}{\ensuremath{\mbox{disk}}}

\newcommand {\dist}{\ensuremath{\mbox{dist}}}

\begin{document}

\maketitle

\begin{abstract}
    In the \AGP/ 
    we are given a polygon $P\subset [0,L]^2$ on $n$ vertices and a number $k$.
    We want to find a guard set $G$ of size $k$,
    such that each point in $P$ is \emph{seen}
    by a guard in $G$.
    Formally, a guard $g$ sees a point 
    $p \in P$ if the line segment $pg$ is fully contained inside $P$.
    
    The only, currently known, correct algorithm 
    to solve the \AGP/ exactly, uses algebraic methods~\cite{EfratH02,EfratH06}. 
    Due to the recent result that the \AGP/ 
    is \ER/-complete~\cite{ARTETR},
    it seems unlikely that algebraic methods can be avoided 
    for any exact algorithm.
    On the other hand, any example that requires irrational
    coordinates, or has a unique solution is {\em vulnerable}, 
    to small perturbations, i.e. looses that property
    after a random perturbation. Furthermore, it took more than 
    four decades to come up with an example where 
    irrational coordinates are required~\cite{abrahamsen2017irrational}. 
    Lastly, there is a series of papers that
    implement algorithms that give optimal solutions 
    to medium sized simulated
    instances ($5000$ vertices)~\cite{engineering}.
    The history and practical findings therefore
    indicate that irrational coordinates 
    are a ``very rare'' phenomenon to find 
    in an optimal solution.
    In this paper we give a theoretical explanation. 
    
    Next to worst case analysis, \SmoothAna/ gained popularity
    to explain the practical performance of algorithms, even 
    if they perform badly in the worst case. 
    \SmoothAna/ is an interpolation between average case analysis
    and worst case analysis. 
    The idea is to study the expected performance on small perturbations 
    of the {\em worst} input.
    The performance is measured in terms of the magnitude $\delta$ of 
    the perturbation and the input size.
    We consider four different models of perturbation.
    ({\bf (a)} \MinkInfl/, {\bf (b)} \EdgeInfl/, {\bf (c)} \EdgePert/, and {\bf (d)} \VertexPert/.)
    We show that the expected number of bits to describe 
    optimal guard positions equals 
    %
    \[ {\bf (a),(b),(c)} \ O\left(   \log\left(\frac{nL}{\delta }\right)     \right) , 
    \quad  {\bf (d)} \ O\left(   \log\left(\frac{nL}{\delta \beta}\right) \right),\]
    per guard,
    where the value $\beta$ denotes the minimum
    of the interior and exterior angle of the polygon.
    This shows from a theoretical perspective
    that rational guards with small bit-complexity
    are \emph{typical}.
    Note that describing the guard position is the bottleneck
    to show \NP/-membership.
    To illustrate our findings more, we consider a \emph{discrete}
    model of \EdgeInfl/. We can show that the resulting 
    instances can be solved with high probability in 
    non-deterministic polynomial time, by the {\em \NaiveAlgo/}.
    The \NaiveAlgo/ guesses the correct guard positions 
    and checks if they are guarding the entire polygon.
    As a simple corollary, we can show that there is 
    an algorithm that solves the \AGP/ in expected non-deterministic 
    polynomial time on a \TuringMachine/, however, that algorithm uses 
    algebraic methods with small probability.
    In a {\em continuous} model of perturbation, we show
    that the \NaiveAlgo/ runs in expected non-deterministic 
    time, but spends and additional $O(n^4)$ time on a
    \emph{\realRAM/}, in order to check that the guards are correct.
    Avoiding the \realRAM/ is impossible for continuous 
    perturbations as the coordinates of the vertices
    of the resulting polygons are given
    by real numbers.
    As a corollary, we show correctness of an approximation algorithm
    in the smoothed analysis setting.
    
    The significance of our results is that algebraic methods are
    not needed to solve the \AGP/ in {\em typical} instances.
    This is the first time an \ER/-complete problem
    was analyzed by \SmoothAna/.
\end{abstract}

\section{Introduction}
\paragraph{Definition of the \AGP/}
    In the \AGP/ 
    we are given a polygon $P$ and a number $k$.
    We want to find a guard set $G$ of size $k$,
    such that each point in $P$ is \emph{seen}
    by a guard in $G$.
    Formally, a guard $g$ sees a point 
    $p \in P$ if the line segment $pg$ is fully contained inside $P$.
    We usually denote the vertices of $P$
    by $v_1,\ldots,v_n$, and 
    the number of vertices by $n$.


\subsection{Brief history of the \AGP/}
\label{sec:History-Art}
 The \AGP/ has been one of the core problems in Computational
 Geometry ever since its introduction by Victor Klee in 1973.
 Much of its popularity may stem from the metaphor,
 referring to an actual gallery, with actual guards.
 However, bear in mind
 that the problem is, although motivated by practice,
 a theoretical one about visibility and geometric set-cover.
 One of the earliest results 
 states that any polygon on $n$
 vertices can be guarded by $\lfloor n/3 \rfloor$ guards,
 and sometimes that many 
 guards are needed~\cite{chvatal1975combinatorial}.
 A beautiful five line proof of this fact was given by Steve Fisk~\cite{Fisk78a}.
 
 A significant amount of the literature on the \AGP/
 studies variants of the problem, but
 to survey these results lies outside the scope of this paper.
 Variants can be defined by restricting the position of the
 guards, restricting the shape of the polygon, and 
 altering the notion of visibility. The majority of results either give 
 combinatorial results  or study algorithmic aspects.
 Recently, Bhattacharya, Ghosh and Pal announced a breakthrough,
 where they reported a constant factor-approximation algorithm
 for the case when the guards are restricted to lie on the vertices~\cite{PritamConstantFactor}.
 For further references we recommend the survey and books on the \AGP/ 
 and visibility problems in general~\cite{o1987art, 2004handbook}.

Regarding the classical variant, surprisingly little is known.
In 1979 appeared a linear time algorithm
to guard a polygon with one guard~\cite{LeePreparataoptimal},
and an $O(n^4)$ algorithm to guard a polygon 
with two guards appeared
in the 1991 Master thesis of Belleville~\cite{belleville1991computing,BellevilleCCCG}.
Up until 2002 it was not even known whether the 
\AGP/ is decidable. Decidability of the 
\AGP/ was pointed out by Micha Sharir,
who showed how to encode the \AGP/
into the first order theory of the reals~\cite[See Acknowledgment]{EfratH02}. 
As there exist algorithms to decide the first order theory of
the reals, the \AGP/ is decidable~\cite{basu2006algorithms}.
In particular, the algorithm runs in $n^{O(k)}$ time,
where $n$ is the number of vertices and $k$ is 
the number of guards.
(Here, we assume that every vertex is encoded using 
$O(\log n)$ bits.)
It was also shown that the \AGP/ is W[1]-hard,
and not solvable in time $f(k)n^{o(k/\log k)}$, for any 
computable function $f$ unless
the Exponential Time Hypothesis fails~\cite{BonnetM16}.
(The Exponential Time Hypothesis states essentially that
$3$SAT on $N$ variables cannot be solved in $2^{o(N)}$.)
This indicates that the running time of Sharir's algorithm
is best possible, but it does not resolve the question,
whether algebraic methods are required.

The first approximation algorithm was established by 
Bonnet and Miltzow in 2016. Using some mild assumptions,
they showed that a fine grid contains a solution that
is at most a factor $9$ away from the optimum~\cite{BonnetM17Approx}.
Combined with a result of Efrat and Har-Peled~\cite{EfratH06}, this yielded
the first approximation algorithm for the \AGP/.

From the complexity perspective the problem was studied 
almost exclusively from the perspective of
\NP/-hardness and inapproximability~\cite{LeeLin86,SchuchardtH95,eidenbenz2001inapproximability,aggarwal1984art},
while the question of NP-membership was rarely asked.
A first doubt of \NP/-membership was raised in 2017, 
when Abrahamsen, Adamaszek and Miltzow
showed that
there exists polygons with vertices given by integer 
coordinates, that can be guarded by three guards, 
in which case some guards must necessarily have 
irrational coordinates~\cite{abrahamsen2017irrational}. 
(It is an open problem whether irrational guards may be required 
for polygons which can be guarded by two guards.)

Shortly after,  the same authors could show that
the \AGP/ is complete for the complexity class \ER/~\cite{ARTETR}.
\ER/-completeness can also be stated as follows. The 
\AGP/ is equivalent, under polynomial
time reductions, to deciding whether a system of polynomial
equations has a solution.
This implies that it seems unlikely that the algebraic algorithms
used by Sharir to solve the \AGP/ can be avoided,
to give optimal solutions in the \emph{worst case}.
We will discuss the complexity class \ER/ later in more detail. 
Assuming \NP/ $\neq$ \ER/, we can conclude that 
the \AGP/ is not contained in~\NP/.
 
While those theoretical results are quite negative, the history
and practical experiences tell a more positive story.
First of all, it took more than four decades before an example
could be found that requires irrational guards~\cite{abrahamsen2017irrational}.
Moreover, that example is highly vulnerable to small
perturbations, meaning that small perturbations
lead to a polygon that can be guarded optimally by 
guards with rational coordinates. 
Regarding the practical study of the \AGP/, we want 
to point out that several researchers have implemented heuristics,
that were capable of finding optimal solutions
for a large class of simulated 
instances~\cite{engineering, PracticalARTMasterFriedrich,
PracticalARTbottino2011, PracticalARTkroller2012, PracticalARTcouto2011,
PracticalARTbottino2008, PracticalARTcouto2008, PracticalARTamit2010, PracticalBorrmann}.
Even up to $5000$ vertices.

 Let us point out that several researchers have asked whether the \AGP/ 
 requires irrational
 coordinates~\cite{SandorEmails, dagstuhlSeminar,OpenProblem,engineering,DiscretizeTerrain}. 
 While Abrahamsen, Adamaszek and Miltzow~\cite{ARTETR,abrahamsen2017irrational}
 gave a negative answer for the worst case scenario, we 
 give a positive answer for {\em typical} instances.
 
We summarize that there is a large discrepancy between
the theoretical findings that the \AGP/ is \ER/-complete and the practical observation
that there {\em usually} exists an optimal solution
with rational coordinates. Our results explain this discrepancy.

\subsection{Background on 
the Existential Theory of the Reals: \ER/}
\label{sec:ETR}

A new complexity class began to emerge around the existential 
theory of the reals in the 1980s \cite{shor1991stretchability}, 
and more recently the notation \ER/ was introduced 
along with the formal 
definition~\cite{schaefer2013realizability}.\footnote{Previously, 
ETR (Existential Theory of the Reals) had been used ambiguously 
to refer to the formal language, the corresponding decision problem, or its algorithmic complexity.}
The class \ER/ is the class of all decision  
problems that are many-one reducible in polynomial time to deciding
whether a given polynomial $Q\in \Z[x_1,\ldots,x_n]$
has a real root, i.e. a solution $x\in\R^n$ such that $Q(x) = 0$.
From the field of real algebraic geometry~\cite{basu2006algorithms},
we know that \[\NP/ \subseteq \ER/ \subseteq \PSPACE/ .\]
At its core, hardness for \ER/ provides an explanation for
why some problems may not lie in \NP/. 

One of the most famous algorithmic questions is the recognition
of {\em segment intersection graphs}. Here, we are given a graph
and we are asked to decide whether there exists a set of segments which 
represent the graph in the following way: every segment represents a vertex of the graph and two segments intersect if and only if 
their corresponding vertices are adjacent.
It is easy to believe that the problem lies in \NP/.
By a simple perturbation argument, it can be assumed
that all coordinates of segment endpoints 
are represented by integers, and as long 
as the number of bits needed to represent them can be bounded by
some polynomial, we would be done. 
Indeed, Matou{\v{s}}ek~\cite{DBLP:journals/corr/Matousek14} 
comments that 
\vspace{0.3cm}
\begin{adjustwidth}{1cm}{1cm}
``Serious people seriously conjectured that the number of digits can be polynomially bounded---but it cannot.''
\end{adjustwidth}
\vspace{0.3cm}
Indeed McDiarmid and 
Müller have shown that 
the number of bits needed
to represent certain families of segment intersection graphs
is at least $2^{\Omega(n)}$~\cite{mcdiarmid2013integer}.
Similar large coordinate phenomena have been observed also
for other geometric problems~\cite{goodman1990intrinsic}.
However, this does not exclude that those 
respective problems lie in \NP/, as it may be possible 
to describe a different certificate of polynomial size.
Indeed recognition of string graphs is in \NP/~\cite{schaefer2003recognizing, PachString,PachString-GD, Schaefer-String-STOC}, although
there are families which require an {exponential 
number of crossings} in any string representation~\cite{kratochvil1991string}. 

The complexity class \ER/ provides a tool to give much 
more compelling arguments that a problem may not lie in \NP/
than merely observing that the naive way of placing
the problem into \NP/ does not work.
Indeed various problems have been shown to be 
\ER/-complete~\cite{AnnaPreparation,AreasKleist,shitov2016universality, richter1995realization,garg2015etr,schaefer2013realizability, cardinal2017intersection, cardinal2017recognition, kang2011sphere}
and thus either non of them lie in \NP/ or all of them do.

Another important aspect of \ER/-complete problems is that
we have no chance of solving even small instances (with the current methods), 
as the constants in the known algebraic algorithms are too large.

Most relevant in our context is the \ER/-completeness
that was shown by Abrahamsen, Adamaszek and Miltzow in
2017~\cite{ARTETR}. 
In 1987, O'Rourke~\cite[page 232]{o1987art} commented 
on the \NP/-membership in his famous
book on the \AGP/ as follows:
\vspace{0.3cm}
\begin{adjustwidth}{1cm}{1cm}
``The usual first step in a proof of \NP/-completeness is to 
show that the problem is a member of the class of \NP/ problems, 
that is, solvable via a non-deterministic algorithm in 
polynomial time 
(\ldots).
Often this is 
easy, merely requiring a demonstration that a solution ``guessed" 
by a non-deterministic program can be checked in polynomial time.
(\ldots)
however, it is unclear how to establish this.''
\end{adjustwidth}
\vspace{0.3cm}
The \ER/-completeness implies that there is no
algorithm, which runs in polynomial non-deterministic time 
and can solve the \AGP/ always, unless $\NP/ = \ER/$. 
This result is our main motivation, to see if there is a 
simple algorithm that solves the \AGP/. As we don't expect 
that such an algorithm is correct in the worst case, 
we turn our attention to different ways to analyze algorithms. 

\subsection{\SmoothAna/}
\label{sec:Smoothed-Analysis}
Some algorithms perform 
much better than predicted by their
worst case analysis. The most famous example seems
to be the Simplex-Algorithm. It is an algorithm
that solves   
linear programming efficiently in practice, although it is known that there
are instances for seemingly all variants of
the algorithm that take an exponential amount 
of time (see for instance~\cite{klee1970good}). 
There are several possible 
ways to explain this behavior. For example, it could be 
that all practical instances have some structural 
properties, which we have not yet discovered. 
We could imagine that a more clever analysis of 
the Simplex-Algorithm would yield that it runs 
in polynomial time, assuming the property 
is presented. To the best of our knowledge such 
a property has not yet been identified.
Another approach would be to argue that worst 
case examples are just very ``rare in practice''.
The problem with this approach is that it is
difficult to formalize. 

\paragraph{Average Case Analysis}
The first approach is to assume that a 
``practical instance'' is drawn uniformly 
at random from all possible instances.
This is also called average case analysis.
The first problem with this approach is that
we have to choose a probability space.
The second problem is that practical instances
are often highly structured. For instance,
planar graphs play an important role in many situations,
but almost never appear if each edge is drawn randomly
with probability $1/2$.
The third problem is that it only says that a
big portion of the probability space behaves nicely,
but there might still be a big region in the 
space of instances which is really bad.
And maybe that region is actually of practical relevance.

\paragraph{\SmoothAna/}
The second approach is a nice combination of the 
average case and the worst case analysis
and generally referred to as \SmoothAna/,
as it {\em smoothly} interpolates between the two. 
It was developed by Spielman and Teng~\cite{spielman2004smoothed}, who introduced the field in 
their celebrated seminal paper ``\SmoothAna/ of algorithms: 
Why the simplex algorithm usually takes polynomial time".
Both authors received the
G\"{o}del Prize in 2008, 
and the paper was one of the winners of the Fulkerson Prize
in 2009. In 2010 Spielman received
the Nevanlinna Prize for developing \SmoothAna/.

The smoothed expected running time can be defined as follows:
Let us fix some $\delta$, which describes the maximum 
magnitude of perturbation. 
We denote by $(\Omega_\delta,\mu_\delta)$ a 
corresponding probability space where 
each $x\in \Omega_\delta$ defines for each instance $I$  
a new `perturbed' instance 
$I_x$.
We denote by $T(I_x)$, 
the time to solve the instance $I_x$.
Now the smoothed expected running time of instance $I$ equals
\[ T_\delta (I) = \mathop{\E}_{x \in \Omega_\delta} \, T(I_x) = \int_{x\in \Omega_\delta} T(x) \mu_\delta(x).\]
If we denote by $\Gamma_n$ the set of instances of size $n$,
then the smoothed running time equals:
\[T_{\mbox{smooth}}(n) = \max_{I\in \Gamma_n} \, \mathop{\E}_{x\in \Omega_\delta} T(I_x). \]
Roughly speaking this can be interpreted
as saying, that not only do the majority
of instances have to behave nicely, but actually
in every neighborhood the majority of instances
behave nicely.
The expected running time is measured in terms of $n$ and $\delta$.
If the expected running time is small in terms of $1/\delta$
then this  means that difficult instances are \emph{fragile}
with respect to perturbations.
This serves as theoretical explanation
why such instances may not appear in practice.

It is not a priori clear how this probability 
space $\Omega_\delta$ should be defined.
It is easy, if our object is a point in Euclidean
space and proximity can be defined by Euclidean distance,
but less obvious if your object is a permutation or
a polygonal region in the plane.

Although the concept of \SmoothAna/ is more complicated
than simple worst case analysis, it is a new success story in
theoretical computer science. 
It could be shown that various algorithms actually 
run in polynomial time, explaining very well their practical 
performance.

One of the highlights is an analysis of 
the Nemhauser-Ullmann Algorithm~\cite{nemhauser1969discrete}
for the knapsack problem running in 
smoothed polynomial time~\cite{KnapsackSmooth}.
Those results could be generalized, yielding that
every binary optimization problem can be solved
in smoothed polynomial time if and only if it can be solved
in pseudopolynomial time~\cite{beier2006typical}.
Other famous examples are the \SmoothAna/ 
of $k$-means algorithm~\cite{arthur2006worst}, 
the $2$-OPT TSP local
search algorithm~\cite{englert2007worst}, and the 
local search algorithm for MaxCut~\cite{MaxCUTsmoothed}.
Recently, the smoothed number of faces on the convex hull of a point set
was analyzed under Gaussian noise~\cite{devillers2016smoothed}.
The currently best analysis of the Simplex Algorithm
yields a running time of $O(d^2 \sqrt{\log n}\, \sigma^{-2} + d^5 \log^{3/2}n)$,
with $d$ the number of variables, $\sigma^2$ the variance of Gaussian noise, 
and $n$ the number of constraints~\cite{dadush2018friendly}.
We refer the interested reader to surveys, mini-courses,
and lecture notes~\cite{manthey2011smoothed, SurveySmoothTim,
smoothCourseHeiko, smoothCourseTim}.

\subsection{Defintions}
\label{sec:Definitions}

    We now define the models of perturbation
    which we find relevant for the \AGP/. See 
    Figure~\ref{fig:Overview-Perturbation-Models} for an illustration. 
    We will also explicitly state 
    our assumptions on the underlying polygons.
    For all our probability spaces $\Omega_\delta$, we 
    are using the uniform distribution.
    
    \begin{figure}[tbh]
        \centering
        
        \begin{subfigure}[t]{0.4\textwidth}
            \centering
            \includegraphics[page = 2]{Edge-Inflation.pdf} 
            \caption{The polygon together with an \EdgeInfl/.}
            \label{fig:Edge-Inflation}
        \end{subfigure}
        \hspace{0.1cm}
       \begin{subfigure}[t]{0.4\textwidth}
            \centering
            \includegraphics{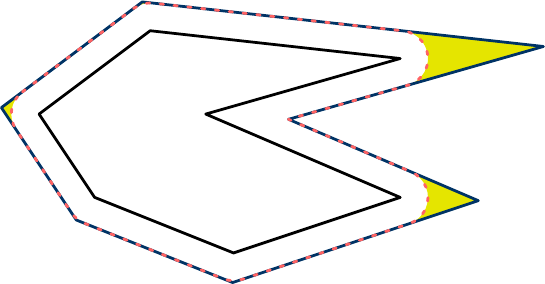}
            \caption{If we continue the edges, of a \MinkInfl/, 
            we get an \EdgeInfl/.}
            \label{fig:Edge-Minkowski-Difference}
        \end{subfigure}
       
        \begin{subfigure}[t]{0.3\textwidth}
            \centering
            \includegraphics[page = 2]{Minkowski-Inflation.pdf}    
            \caption{The Minkowski-sum of a polygon 
            together with a disk.}
            \label{fig:Minkowski-Inflation}
        \end{subfigure}
        \hspace{0.1cm}
        \begin{subfigure}[t]{0.3\textwidth}
            \centering
            \includegraphics[page = 2]{Vertex-Perturbation.pdf} 
            \caption{A \VertexPert/.}
            \label{fig:Vertex-Perturbation}
        \end{subfigure}
        \hspace{0.1cm}
        \begin{subfigure}[t]{0.3\textwidth}
            \centering
            \includegraphics[page = 2]{Edge-Perturbation.pdf}    
            \caption{A random \EdgePert/. 
            Note that all angles of edges are preserved.}
            \label{fig:Edge-Perturbation}
        \end{subfigure}
        
        \caption{Overview, over various models, how a polygon can be perturbed.}
        \label{fig:Overview-Perturbation-Models}
    \end{figure}
    \paragraph{\VertexPert/}
    Given a polygon $P$ on $n$ vertices $v_1,\ldots,v_n$ and a vector $u = (u_1,\ldots,u_{2n})  \in \R^{2n}$,
    we denote by $P_u$ the polygon with vertices $v_i'= v_i + (u_{2i-1},u_{2i})$.
    We say that $u\in \R^{2n}$ represents a {\em \VertexPert/} of magnitude
    $\delta$, if for every $i$ holds that $u_{2i-1}^2 + u_{2i}^2 \leq \delta^2$.
    See Figure~\ref{fig:Vertex-Perturbation} for an illustration.
    We define the corresponding probability space as 
    $\Omega_\delta = \disk(\delta)^n$.
    (We define $\disk(\delta)= \{ x\in \R^2 : x_1^2 + x_2^2 \leq \delta^2 \}$.)
    
    \paragraph{\MinkInfl/}
    Given a polygon $P$ on $n$ vertices $v_1,\ldots,v_n$ and a number $t\in \R$,
    we denote by $P_t = P \oplus \disk(t) \subseteq \R^2$ the 
    \MinkInfl/ of $P$ by $t$, see Figure~\ref{fig:Minkowski-Inflation}. 
    We say the magnitude of the inflation is at most $\delta$ in case that $t\leq \delta$.
    We define the corresponding probability space as 
    $\Omega_\delta = [0,\delta]$.
    Recall that given two sets $A,B \subset \R^2$, the Minkowski-sum
    is defined as $A\oplus B = \{x+y\in \R^2 : x\in A \textrm{ and } y\in B\}$.

\paragraph{\EdgeInfl/}
    Let $P$ and $P'$ be polygons on the same number of vertices.
    We say edge $e$ is {\em shifted} in polygon $P'$, if there is a 
    corresponding edge $e'$ in $P'$ and $e$ and $e'$ are parallel.
    The shift has {\em magnitude $t$}, if the distance of the supporting lines of 
    $e$ and $e'$ is $t$.
    We distinguish between a shift to the {\em outside} and to
    the {\em inside} in the obvious way.
    We say $P_t$ is an {\em \EdgeInfl/} of $P$ if every edge of
    $P$ is shifted by $t$ to the outside in $P_t$.
    See Figure~\ref{fig:Edge-Inflation}.
    We say the magnitude of the inflation is at most $\delta$ in case that $t\leq \delta$.
    We define the corresponding probability space as 
    $\Omega_\delta = [0,\delta]$.
    
    We also define a {\em discrete probability space}
    $\Omega_{\delta,q} = \{\frac{i\delta}{q}: 0\leq i \leq q-1\}\subseteq [0,\delta]$.
    Here $q\in \N$ indicates the \emph{granularity} of the perturbation.
    Furthermore, we assume $\delta \in \Q$.
    Consider the example of the triangle $\Delta$ on the vertices $(0,0), (1,0), (0,1)$.
    Note that for any rational number $r\in \Q$ holds that the
    \EdgeInfl/ $\Delta_r$ does not have rational coordinates.
    However, it is easy to give an explicit description of
    the coordinates, using radicals, i.e., square-roots.

\paragraph{\EdgePert/}
    Given a polygon $P$ and  $(t_1,\ldots,t_n) = v\in [-\delta,\delta]^n = \Omega_\delta$,
    we define the \EdgePert/ $P_v$ by 
    shifting the $i$-th edge of the polygon $P$
    by $t_i$, either inside or outside.
    (Depending on the sign of $t_i$.)
    See Figure~\ref{fig:Edge-Perturbation} for an illustration.

    \paragraph{Pointedness}
    Given a polygon $P$ we define the \emph{pointedness}
    $\beta = \beta(P)$  as follows: 
    Let $\alpha$ be the smallest interior or exterior angle of $P$,
    then $\beta = \alpha/8$. 
    We assume that the interior or exterior angle does not change by more
    than a factor of $2$ by any perturbation considered in this paper. 
    This can be achieved by having the perturbation sufficiently small.
    As a consequence of this assumption every interior and exterior
    angle will be always at least $4\beta$ for any \VertexPert/
    that we consider.

\paragraph{The \NaiveAlgo/}
    The \NaiveAlgo/ is guessing
    non-deterministically those coordinates, with a minimum 
    number of bits needed to describe an optimal guard set.
    Thereafter the \NaiveAlgo/ checks if those guessed
    guards are indeed guarding the entire polygon.
    The \NaiveAlgo/ will fail in case that
    there is no way to guard the given polygon optimally
    by guards whose coordinates can be described by rational numbers.
    In that case, we say that the running time is infinite.
    
\paragraph{Assumptions and Model of Computation}
    In this paper we are working with two different
    models of computation, which is a consequence of
    the model of perturbation that we are using.
    
    In the first model, we are using a continuous model
    of perturbation. In that case, the polygon is 
    described by real numbers. Consequently, we assume
    that we can do computations on the \realRAM/.
    It is crucial however that
    the non-deterministic part of the computation cannot guess 
    real numbers, but only rational numbers. 
    We measure the non-deterministic time of the algorithm
    in terms of the number of bits that are needed to represent
    the rational numbers that are guessed by the algorithm.
    
    In the second model, we are using a discrete model
    of perturbation. 
    In this setting, we are assuming that the vertices
    of the input polygon can be described by rational numbers.
    However, this does not lead always
    to rational vertices of the perturbed
    polygon, as we pointed out above. 
    Nevertheless, as we will argue, 
    we can do all basic operations on a \wordRAM/.
    
    Otherwise we are not making any assumptions
    on the polygons themselves. In particular, the polygon may have holes. 
    However, note that if the perturbation is
    too large in comparison to the polygon, then the 
    perturbed polygon may not be well-defined.
    This does not cause any serious issues since \SmoothAna/ usually applies in
    settings where the perturbation is thought of as 
    relatively small. So without defining this formally,
    we  will throughout the entire paper assume that the perturbation is
    sufficiently small, such that the perturbation defines a polygon.
    In case of \MinkInfl/, the output region is \emph{not} 
    a polygon anymore. 
    We can still reason about it, by making the assumption that
    we can do basic operations (like computing visibility
    regions and determine if a point is inside the regions, etc.).
    Although, the \MinkInfl/ is an important technical step
    for us, it is not the most important perturbation model.
    Thus, we will not discuss how realistic these assumptions are.

\subsection{Results}
\label{sec:Results}

Our main result states that typical instances do not
require irrational guards and the expected 
number of bits per guard is logarithmic.
The result establishes that algebraic methods
are {\em not} needed in typical instances.

\begin{restatable}[Bit-complexity]{theorem}{UpperBits}
\label{thm:MainBits}
    Let $P$ be a polygon, suppose
    $P \subset [0,L]^2$ for some positive integer $L$, and let 
     $\beta$ denote the pointedness of~$P$. 
    If $\delta>0$ is the magnitude of~a 
    
    \quad {\bf (a)} \MinkInfl/   \quad {\bf (b)} \EdgeInfl/
    \quad {\bf (c)} \EdgePert/ \quad {\bf (d)} \VertexPert/,
    
    \noindent then the  expected number of bits per guard to describe
    an optimal solution equals
    
    \quad {\bf (a),(b),(c)} $O\left( \log\left(\frac{nL}{\delta }\right)     \right)$  
    \quad  {\bf (d)} $O\left( \log\left(\frac{nL}{\delta \beta}\right) \right)$.
\end{restatable} 

As a simple corollary of the proof, 
we get that a fine grid of expected 
width $w = 2^{O(\log(nL/\delta))} = (nL/\delta)^{O(1)}$
will contain an optimal guarding set.
This may appear 
at first sight as a {\em candidate set} of polynomial
size, however recall that the vertices are given
in binary and thus $L$ may be exponential in 
the input size.

Theorem~\ref{thm:MainBits} leads immediately to
a non-deterministic polynomial time algorithm for the
\AGP/ in the \SmoothAna/ model.
However, we still need the \realRAM/ to test
if the guard positions are correct.
\begin{restatable}[Expected \NP/-time]{theorem}{Upper}
\label{thm:MainUpper}
    Let $P$ be a polygon,
    suppose $P \subset [0,L]^2$ for some positive integer $L$, and let 
     $\beta$ denote the pointedness of $P$. 
    If $\delta>0$ is the magnitude of~a 
    
    \quad {\bf (a)} \MinkInfl/   \quad {\bf (b)} \EdgeInfl/
    \quad {\bf (c)} \EdgePert/ \quad {\bf (d)} \VertexPert/,
    
    \noindent then the \NaiveAlgo/ runs in expected
    
    \quad {\bf (a),(b),(c)} $O\left(  n \log\left(\frac{nL}{\delta }\right)     \right)$  
    \quad  {\bf (d)} $O\left(  n \log\left(\frac{nL}{\delta \beta}\right) \right)$
    
    non-deterministic time.
    Furthermore, the algorithm 
    takes an additional $O(n^4)$ deterministic time
    on a \realRAM/.
\end{restatable} 

To avoid the \realRAM/, we switch to a discrete model
of perturbation and focus solely on the \EdgeInfl/.
Note that we know that some instances of the \AGP/
require irrational guards and the \NaiveAlgo/ 
would have an infinite running time. 
Thus we cannot expect to acquire a finite expected
running time over a finite probability space.

\begin{restatable}[\NP/-time With High Probability]{theorem}{UpperProbability}
\label{thm:MainHighProb}
    Let $P$ be 
    a polygon and suppose
    $P \subset [0,L]^2$ for some positive integer $L$.
    If $\delta>0$ is the magnitude of~a {\bf discrete} \EdgeInfl/
    with granularity $q$ sufficiently large
    ($q>\frac{2n}{p}$ is sufficient),
    then the \NaiveAlgo/ runs with probability
    $1-p$
    in $O(n\log \frac{L n}{\delta p})$ 
    non-deterministic time
    and some additional polynomial deterministic time
    on an ordinary \TuringMachine/.
\end{restatable} 

    Note that the dependence on $n,L,\delta$ and $\beta$ are
    in each case logarithmic. This is unusual compared
    to other running times derived in \SmoothAna/.
    Usually the dependence is $(1/\delta)^c$, where $c$
    is some constant, which is not always very small.
    The dependence on $1/\delta$ can be interpreted as how
    fragile the hard instances are to perturbations, 
    and here, a logarithmic dependence indicate that they
    are exponentially more fragile than in other settings where \SmoothAna/ has been applied.

    More important than the analysis itself, 
    is the fact that 
    this is the first time any exact algorithm, without
    using algebraic methods, could be shown to work correctly
    in a generally accepted theoretical model.
    We hope that our analysis will inspire researchers to
    find ways to show that also other (more practical) 
    algorithms are correct
    in the \SmoothAna/ model or a different model.
    
    If we choose $p$ from Theorem~\ref{thm:MainHighProb} 
    sufficiently small, only a small fraction of 
    the instances cannot be solved by the \NaiveAlgo/.
    If we use algebraic methods for those instances,
    we easily get an algorithm that runs in expected 
    non-deterministic polynomial time.
    \begin{restatable}[\NP/-time]{theorem}{NPTime}
        \label{thm:NPtime}
    Let $P$ be a polygon and suppose
    $P \subset [0,L]^2$ for some positive integer $L = n^{O(1)}$.
    If $\delta>0$ is the magnitude of~a {\bf discrete} \EdgeInfl/
    with granularity $q$ sufficiently large
    ($q>n^{\Omega(n)}$ is sufficient),
    then there is an algorithm that 
    runs in expected non-deterministic polynomial time
    on an ordinary \TuringMachine/.
    \end{restatable} 

    Note that although \SmoothAna/ is a fairly involved
    concept, the proofs are relatively simple. 

    It is tempting to think that one could improve the result
    by scaling the polygon such that it fits into a unit square.
    Note that in this case one also has to scale
    down the perturbation by the same magnitude and one obtains the same 
    result.
    
    As a corollary, we can see that a fine grid contains
    an optimal solution in the \SmoothAna/ setting.
    Efrat and Har-Peled gave an approximation algorithm 
    to a variant where guards are restricted to a grid~\cite{EfratH06}.
    We repeat their theorem, reformulated and simplified slightly.
    Let $OPT(P)$ be an optimal guarding set of the polygon $P$, and let
    $OPT(P, C)$ denote the optimal way to guard the polygon $P$,
    under the restriction that all guards need to lie 
    on points in $C$.
    (We will only consider sets $C$, which are guarding $P$.)
    We say a polygon is {\em simple}
    if it does not contain any hole.
    \begin{thmNoNumber}[Efrat \& Har-Peled~\cite{EfratH06}]
        Given a simple polygon $P$ with $n$ vertices, one can spread
        a grid $\Gamma$ inside $P$, and compute a guard set of size
        $O( |OPT(P,\Gamma) \cdot \log |OPT(P,\Gamma) |)$. 
        The expected running time of the algorithm is
        $O ( n^3 \log^2 n  \log^2 \Delta.)$ , where $\Delta$ is 
        the ratio between the diameter of the polygon and the grid width.
        The algorithm runs on a \realRAM/.
    \end{thmNoNumber}
    Note that the theorem does not assume that $|OPT(P,\Gamma)|$
    and $|OPT(P)|$ are related. Due to this gap, the algorithm
    was not known to be an approximation algorithm of the optimum.
    Bonnet and Miltzow filled this gap, as they showed that a fine grid 
    contains a constant factor
    approximation of the optimum, under some mild assumptions on the polygon.
    So in fact the algorithm of Efrat and Har-Peled provides an $O(\log |OPT(P)|)$-approximation 
    to the classical \AGP/~\cite{BonnetM17Approx}.
    Now the technical demanding proof by Bonnet and Miltzow can be
    replaced by a simple \SmoothAna/.
    Also our analysis yields better constants both
    in the approximation factor and the grid-width.
    Note that $\Delta \approx L/\delta$ in our notation.
    With our notation the new theorem yields:
    \begin{restatable}[Approximation-Algorithm]{theorem}{Approx}
    \label{thm:MainApprox}
        Let $P$ be a simple polygon, suppose
        $P \subset [0,L]^2$ for some positive integer $L$,
        and let $\delta$ denote the magnitude of an \EdgeInfl/.
        Then the  algorithm by Efrat and Har-Peled  
        runs in expected time $O( n^3 \log^2 n  \log^2 (n L/\delta))$, and gives a solution
        which is within a factor $O(\log |OPT(P_t)|)$ from the optimum number 
        of guards for $P_t$. The algorithm runs on a \realRAM/.
    \end{restatable}
    
    Note that the expectation is here both about the 
    randomness used by the algorithm and the randomness
    coming from the \SmoothAna/.

    Our methods mainly rely on a simple analysis
    of \EdgeInfl/. It seems to be the case that any analysis
    that works out for \EdgeInfl/ is likely to also work for 
    other models of perturbation.
    In Section~\ref{sec:DiscussPertubations},
    we will discuss why it is (arguably) better
    to focus on \EdgeInfl/.
    Nevertheless, the relative simplicity of analyzing \EdgeInfl/ is 
    the main reason why we only consider \EdgeInfl/ in
    Theorem~~\ref{thm:MainHighProb}, 
    Theorem~\ref{thm:NPtime} and Theorem~\ref{thm:MainApprox}.

    \paragraph{Notation}
    We write
    $f(n) \leq_c g(n)$, to indicate
    $f(n) = O(g(n))$ or equivalently 
    $f(n) \leq  c g(n)$, for some large enough constant $c$. 

\section{Preliminaries}
\label{sec:Preliminaries}

In this section we establish some
general facts that will be needed throughout
the paper.

The key idea of the paper are some monotonicity properties of
\MinkInfl/ and \EdgeInfl/. 
Roughly speaking guarding can only get easier after inflations.

\begin{lemma}[Fixed \MinkInfl/]
\label{lem:Fixed-Minkowski-Inflation}
    Let $P$ be a polygon, $t>0$
    and $P_t$ its \MinkInfl/ by magnitude $t$. 
    Then $|OPT(P)|\geq |OPT(P_t, w \Z^2 )|$,
    for any $w\leq \sqrt{2}t$.
\end{lemma} 
\begin{proof}
    Given $OPT = OPT(P)$, we define a set $G\subseteq w \Z^2$ 
    of guards of size $|G| = |OPT|$, by rounding every point
    in $OPT$ to its closest grid point in $w \Z^2$.
    We will show that $G$ guards $P_t$. See to the left of 
    Figure~\ref{fig:Inflations-Proof} for an illustration.

    \begin{figure}[htbp]
    \centering
    \includegraphics{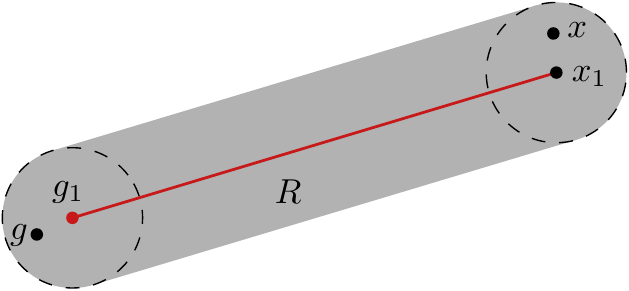}
    \hspace{1cm}
    \includegraphics{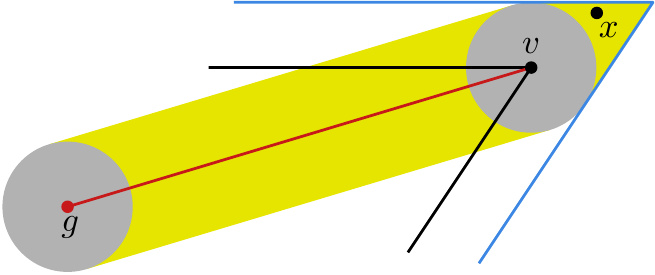}
    \caption{Left: The Region $R$ is convex, and contains a guard $g \in G$ and the point $x$.
    Thus $x$ is guarded by $g$. 
    Right: The Region $R'$ is easily seen to be convex.}
    \label{fig:Inflations-Proof}
    \end{figure}
   
    Let us fix some arbitrary point $x \in P_t$. 
    It is sufficient to show that $G$ guards $x$.
    By definition of $P_t$, there exists an $x_1\in P$
    and an $x_2 \in \disk(t)$ such that $x = x_1 + x_2$.
    Furthermore let $g_1$ be a guard of $OPT$ that guards
    $x_1$. 
    Consider the region $R = g_1x_1 \oplus \disk(t)$,
    i.e., the Minkowski-sum of the segment $g_1x_1$ with a disk of radius $t$.
    As the segment $g_1x_1$ is contained in $P$, 
    it holds that $R$ is contained in $P_t$.
    Also as both the segment and the disk are convex, 
    so is $R$. At last notice that $R$ contains a point $g\in G$,
    as every disk of radius $t$ contains a point of the 
    grid $w\Z^2$ with  $w =\sqrt{2}t$.
    As $R$ is convex, $g \in G$ guards $x$.  
\end{proof}

\begin{lemma}[Fixed Edge Inflation]
\label{lem:Fixed-Edge-Inflation}
    Let $P$ be a polygon with integer coordinates and $t>0$
    and $P_t$ the \EdgeInfl/ of $P$ by $t$. 
    Then $|OPT(P)|\geq |OPT(P_t, w \Z^2 )|$, for any $w\leq \sqrt{2}t$.
\end{lemma}
\begin{proof}
    We follow closely the proof of
    Lemma~\ref{lem:Fixed-Minkowski-Inflation}.
    See to the right of 
    Figure~\ref{fig:Inflations-Proof} for an illustration.
    
    Given $OPT = OPT(P)$, we define a set $G\subseteq w \Z^2$ 
    of guards of size $|G| = |OPT|$, by rounding every point
    in $OPT$ to its closest grid point in $w\Z^2$.
    We will show that $G$ guards $P_t$.
    Note that in an edge inflation by $t$,
    we get the same polygon as by a \MinkInfl/ 
    by $t$, except that we have to add
    some small regions at the convex corners, as illustrated
    in Figure~\ref{fig:Edge-Inflation}.
    We already know that $G$ guards the Minkowski
    $t$-inflation of $P$. So it remains to show that 
    $G$ guards those little extra regions, 
    as discussed above.

    Let us fix some arbitrary point $x \in P_t$
    inside one of those extra regions. 
    We will show that $G$ guards $x$.
    Let $v$ be the vertex according to the region
    that $x$ sits in.
    Furthermore let $g_1$ be a guard of $OPT$ that guards $v$. 
    Consider the region $R = g_1v \oplus \disk(0,t)$.
    We define $R'$ as the region $R$ together with
    the region that $x$ sits in.
    Obviously $x \in R'$ and also there exists
    a point of $G$ in $R$.
    It holds by construction that $R'$ is convex.
    This finishes the proof.
\end{proof}

\section{Approximation Algorithm}
\label{sec:Approx}

To showcase our technique, we will first show the 
correctness of the approximation algorithm
of Efrat and Har-Peled~\cite{EfratH02,EfratH06}.

\Approx* 

\begin{proof}
    Let us assume that there are some numbers 
    $0 = t_0< t_1<\ldots < t_\ell = \delta$
    such that for all $i$ and $s\in [t_{i-1},t_{i})$
    holds that $|OPT(P_s)|$ is constant.
    As $| OPT(P_s)| $ is 
    monotonically decreasing it holds that $\ell \leq n$.
    Note that, if we do an \EdgeInfl/ by $s\in [t_{i-1},t_{i})$,
    we know that a grid of width $w = \sqrt{2}(s-t_{i-1})$
    contains an optimal solution.
    And in this case the algorithm of 
    Efrat and Har-Peled~\cite{EfratH06}
    runs in $O(n^3 \log^2 n \log^2 (L/(s-t_{i-1})))$ time.
    We denote by $\E(T_{i})$ the expected running time
    for $s \in [t_{i-1},t_{i})$.
    We denote $\delta_i = t_{i} - t_{i-1}$.
    Using the definition of the smoothed expected running time
    we get 
    \[
        \E(T_{i}) \leq_c \frac{1}{\delta_i}\int_{s\in[t_{i-1},t_{i})} n^3 \log^2 n \log^2 \frac{L}{s-t_{i-1}} .
    \]
    This gives
    \[    
        \E(T_{i}) \leq_c \frac{1}{\delta_i}\int_{s\in[0,\delta_i)} n^3 \log^2 n \log^2 \frac{L}{s}
        \leq_c \frac{1}{\delta_i} \, 
        \delta_i \, n^3 \log^2 n \log^2 \frac{L}{\delta_i}
         =n^3 \log^2 n \log^2 \frac{L}{\delta_i}.
    \]
    Here, we solved the integral with a standard algebra system.
    Now, we are ready to compute the overall expected running time
    \[
        \E(T) = \frac{1}{\delta}\sum_{i=1,\ldots, \ell} \delta_i  \E(T_{i})
          \leq_c \frac{1}{\delta} \sum_{i=1,\ldots, \ell} \delta_i n^3 \log^2 n \log^2 \frac{L}{\delta_i}
          = \frac{n^3 \log^2 n}{\delta}  \sum_{i=1,\ldots, \ell} \delta_i  \log^2 \frac{L}{\delta_i}.
    \]
    As the function $x\log^2(1/x)$  is concave
    the maximum is attained,
    if $\delta_1=\ldots=\delta_\ell = \delta/\ell$.
    Thus we get
    \[
        \E(T) \leq_c \frac{n^3 \log^2 n}{\delta}   \ell \left[\delta/\ell  \log^2 \frac{L}{(\delta/\ell)}\right]
        = n^3 \log^2 n \log^2 \frac{L}{(\delta/\ell)}
        \leq   n^3 \log^2 n \log^2 \frac{L n}{\delta}. \qedhere
    \]
\end{proof}

\section{Expected Number of Bits}
\label{sec:Bits}
This section is devoted to show the following theorem.

\UpperBits*

We will start by considering inflations as 
they play a special
role for the other cases.

\begin{lemma}
    \label{lem:ExpectedBitInflations}
    Let $P$ be a polygon
    and $\delta$ denotes the magnitude
    of a \MinkInfl/ or \EdgeInfl/.
    We assume furthermore, that the polygon $P$ has $n$ vertices 
    and fits inside a square $[0,L]^2$.
    Then the expected number of bits to describe 
    the optimal guard placement equals
    \[O\left( \log(Ln/\delta)     \right),\]
    per guard.
\end{lemma}

\begin{proof}
    Let us assume that there are some numbers 
    $0 = t_0< t_1<\ldots < t_\ell = \delta$
    such that for all $i$ and $s\in [t_{i-1},t_{i})$
    holds that $|OPT(P_s)|$ is constant.
    As $| OPT(P_s)| $ is 
    monotonically decreasing it holds that $\ell \leq n$.
    We denote by $\delta_i = t_i - t_{i-1}$.

    Note that if the perturbation happens to be $s\in [t_{i-1},t_{i}]$ then 
    a grid of width $w = \sqrt{2}(s-t_{i-1})$ contains an
    optimal solution, see Lemma~\ref{lem:Fixed-Minkowski-Inflation}
    and~\ref{lem:Fixed-Edge-Inflation}. 
    Then the number of bits to describe the solution equals
    $O(\log (L/w)) $ per guard. 
    To see this note that we can use 
    $b= \lceil 1/w \rceil$ as denominator of all
    coordinates and the numerators are 
    upper bounded by $\lceil L/w \rceil$.
    Thus $O(\log (L/w))$ bits suffice.
    Let us denote the number of bits 
    after a perturbation by $s$ as $B(s)$.
    We denote by $\E(B_i)$ the expected number
    of bits for $s\in [t_{i-1},t_{i})$.
    The expected number of bits $\E(B_i)$  can be calculated as 
      \[ 
      \E (B_i) = \frac{1}{\delta_i} \int_{s\in [t_{i-1}, t_i)} B(s)   
      \leq_c \frac{1}{\delta_i} \int_{s\in [t_{i-1}, t_i)} \log(L/(t_i-s))  
      = \frac{1}{\delta_i} \int_{s\in [0,\delta_i]} \log(L/s) .
      \]
    Using some computer algebra system, we get
     \[
      = \frac{1}{\delta_i} \delta_i\left(1 + \log(L/\delta_i)     \right)
      \leq_c   \log(L/\delta_i).
    \]
    We are now ready to compute $\E(B)$.
    \[
        \E(B) = \frac{1}{\delta}\sum_{i=1,\ldots, \ell} \delta_i  \E(B_{i})
        \leq_c \frac{1}{\delta}\sum_{i=1,\ldots, \ell} \delta_i  \log(L/\delta_i).
    \]
    As the function $x\log(1/x)$  is concave
    the maximum is attained,
    if $\delta_1=\ldots=\delta_\ell = \delta/\ell$.
    Thus we get
    \[
          \E(B) \leq_c \frac{1}{\delta}\sum_{i=1,\ldots, \ell} \delta/\ell  \log(L\ell/\delta)
          =   \log(L\ell/\delta)
          \leq_c \log(Ln/\delta). \qedhere
    \]
\end{proof}

Let us now turn to the \SmoothAna/ of \EdgePert/s.
The idea is that the set of all \EdgePert/s can be 
decomposed into
a combination of a new \EdgePert/ together with
a small \EdgeInfl/. 
As we know that \EdgeInfl/s behave nicely,
so will \EdgePert/s.

\begin{lemma}
    \label{lem:Expected-Edge-Perturbation-Time}
    Let $P$ be a polygon
    and let $\delta$ denote the magnitude
    of an \EdgePert/. Furthermore, assume
    that the polygon $P$ has $n$ vertices 
    and fits inside a square $[0,L]^2$.
    Then the  expected number of bits to describe
    a guard placement equals
    $O\left( \log( L n/\delta) \right)$
    per guard.
\end{lemma}

\begin{proof}
    Instead of picking a vector $v$ uniformly at random
    from $\Omega_\delta = [-\delta ,\delta]^n$, we describe another 
    random process, which leads to the same result.
    First, we guess the dimension that takes the minimum
    and the maximum entry. There are $n(n-1)$ 
    possibilities. Let $t = \max v - \min v$ be the 
    difference between the maximum and the minimum.
    The remaining entries are chosen in the interval
    $[0,t]$ uniformly at random. This gives a vector $v'$.
    Thereafter, we pick uniformly at random a shift
    $s \in [0, 2\delta - t]$. 
    The vector $v = v' + s {\bf 1}$.
    It is easy to see that this process is equivalent 
    to choosing $v$ uniformly at random from $\Omega_\delta = [-\delta ,\delta]^n$.
    See Figure~\ref{fig:Edge-Perturbation-Contains-Inflation} for 
    an illustration.
    
    \begin{figure}[htbp]
        \centering
        \includegraphics{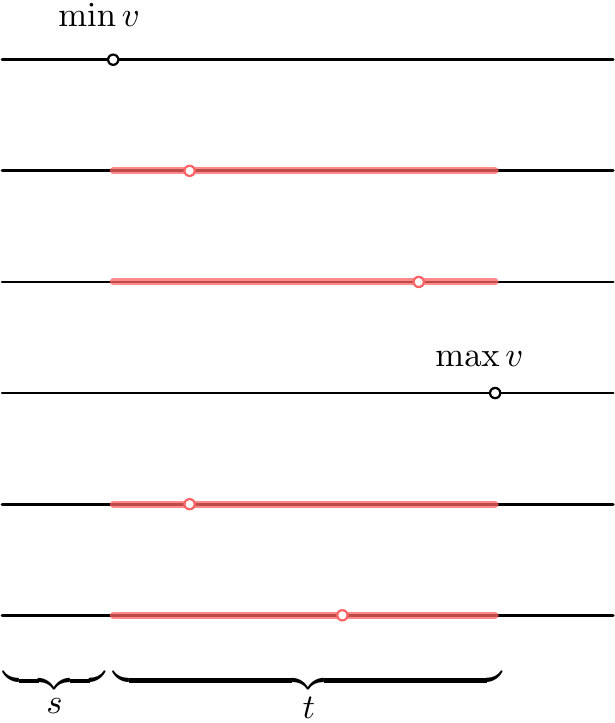}
        \caption{Illustration to pick a random vector $v\in \Omega_\delta = [-\delta,\delta]^n$.
        First guess, where the $\min$ and $\max$ are attained; then 
        guess the difference $\max - \min$; then guess all other entries, and ultimately guess the $\min$.}
        \label{fig:Edge-Perturbation-Contains-Inflation}
    \end{figure}
    
    As in the previous proof, we denote by $B(v)$ the 
    number of required bits per guard to 
    describe an optimal guarding
    on the edge-perturbed polygon $P_v$.
    \begin{align*}
        \E(B) &= \frac{1}{(2\delta)^n} \int_{v\in [-\delta ,\delta]^n} B(v)\\
            &= \frac{1}{(2\delta)^n} n(n-1) 
            \int_{t\in [0,2\delta]} \int_{v'\in [-\delta ,-\delta+t]^{n-2}} 
            \int_{s\in [0,2\delta -t]}B(v' + s{\bf 1})
    \end{align*}
    
    Here, we used the random process as described above.
    Recall $v = v' + s{\bf 1}$. And $v'$ must be padded
    in two dimensions.
    As the inner integral describes a random \EdgeInfl/ of magnitude
    $2\delta - t$, we can use Lemma~\ref{lem:ExpectedBitInflations} to
    get the following upper bound.
    \begin{align*}
            &\leq_c \frac{1}{(2\delta)^n} n(n-1) 
            \int_{t\in [0,2\delta]} \int_{v'\in [-\delta ,-\delta+t]^{n-2}} 
            (2\delta -t)   \log( Ln/(2\delta-t))          
    \end{align*}            
        As the part in the inner integral is independent
        of $v'$, we can easily integrate and get:
    \begin{align}        
            &= \frac{1}{(2\delta)^n} n(n-1) 
            \int_{t\in [0,2\delta]} t^{n-2} (2\delta -t)   \log( Ln/(2\delta-t)) \label{line:integral1}
    \end{align}
    
    In order to compute this integral, we compute 
    the following integral for $a=2\delta$
    and $b = 1/(Ln)$. We will denote the Harmonic numbers
    by $H_n = 1 + \frac{1}{2}+\ldots + \frac{1}{n}$.
    
    \begin{align*}
        & -\int_{t\in [0,a]} t^{n-2} (a-t)   \log(b(a-t))\\
        &= \int_{t\in [0,a]} t^{n-1}    \log(b(a-t)) - a\int_{t\in [0,a]} t^{n-2}    \log(b(a-t))\\
    \end{align*}
        Using some standard
        computer algebra system we get.
    \begin{align*}
        &= \frac{a^n(\log ba -H_n)}{n} - a\frac{a^{n-1}(\log ba -H_{n-1})}{n-1}\\ 
        &= \frac{a^n}{n(n-1)} \left[ {(n-1)(\log ba -H_n)} - {n(\log ba -H_{n-1})}\right]
    \end{align*}

    We simplify the second term as follows.
    Note that $H_n = O(\log n)$.

    \begin{align*}
        & {(n-1)(\log ba -H_n)} - {n(\log ba -H_{n-1})}\\
        &= \log(1/ba) - (n-1)H_{n-1} - \frac{n-1}{n} + (n-1)H_{n-1} + H_{n-1}\\
        &\leq_c \log(1/ba) + \log n \\
        &= \log(n/ba).
    \end{align*}
        
    Now plugging in this simplification in the integral from Line~(\ref{line:integral1}) gives
        
    \begin{align*}
        & \frac{1}{(2\delta)^n} n(n-1) 
            \int_{t\in [0,2\delta]} t^{n-2} (2\delta -t)   
            \log( Ln^2 /(2\delta-t))\\
        &\leq_c \frac{1}{(2\delta)^n} n(n-1) \frac{(2\delta)^n}{n(n-1)} \log(  L n^2 /2\delta)\\
        &= O\left(\log(n L  /\delta) \right).      
    \end{align*}
        
    This shows that the expected running 
    time equals $O\left(\log(n L/\delta) \right)$, as claimed.
\end{proof}

We are now ready to go through the \SmoothAna/
with respect to \VertexPert/. This part follows mainly
the part of \EdgePert/. Again the idea is to think
of a \VertexPert/ as a combination of a \VertexPert/
and an \EdgeInfl/. However, we need an additional trick.
The problem is that it is seemingly impossible to
decompose the space of \VertexPert/ into \EdgeInfl/s
explicitly. Fortunately, it turns out that knowing that 
such a decomposition exists is enough.

\begin{figure}[htbp]
    \centering
    \includegraphics{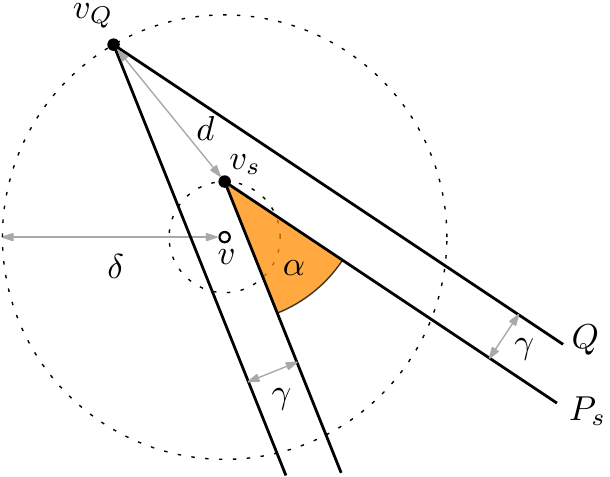}
    \caption{Lower bounding the possible \EdgeInfl/.}
    \label{fig:Max-Edge-Inflation}
\end{figure}

\begin{lemma}\label{lem:pointedness}
    Let $P$ be a polygon
    with pointedness $\beta$
    and $\delta>\delta_0>0$. Furthermore, let $P_s$ be a
    \VertexPert/ of magnitude $\delta_0$, and
    let  $Q = {(P_s)}_t$ be an  \EdgeInfl/ of $P_s$
    of magnitude $\gamma$.
    It holds that for any 
    $\gamma \leq \beta (\delta - \delta_0)$ that $Q$
    is also a \VertexPert/ of $P$ of magnitude $\delta$.
\end{lemma}
\begin{proof}
    See Figure~\ref{fig:Max-Edge-Inflation} for 
    an illustration of the following description.
    Consider a single vertex $v$ of $P$. Let $v_s$ denote the corresponding
    vertex in $P_s$, and $v_Q$ the corresponding vertex
    of $Q$. 
    It suffices to show that 
    $\dist(v,v_Q) \leq \delta$.
    Let $d = \dist(v_Q, v_s)$ and let $\alpha$ denote the interior
    angle at $v_s$. 
    It is easy to see that \[\gamma = d \sin(\alpha/2).\]
    As we remarked above about the assumption on our 
    perturbations, it holds that 
    $ \sin(\alpha/2) \geq \alpha/4 \geq \beta$.
    This implies $\gamma \geq d \beta$
    and in turn $\beta (\delta - \delta_0) \geq \gamma \geq d \beta$,
    by definition of $\gamma$.
    Thus we have $d\leq (\delta - \delta_0)$. 
    This implies $\dist(v,v_Q) \leq \delta_0 + d 
    \leq \delta_0 + (\delta - \delta_0) \leq \delta$.
\end{proof}

    Let $\Omega_\delta = \disk(\delta)^{n}$ be the 
    space that describes all possible \VertexPert/s
    of magnitude $\delta$ of a given polygon $P$.
    As $\delta$ is fixed for the rest of the section,
    we omit it in the notation and just write $\Omega\; (=\Omega_\delta)$.
    In order to define formally the decomposition
    of the $\Omega$ into \EdgeInfl/s,
    we need to define \EdgeDefl/.
    Given a polygon $P$, we say $Q$ is
    an  \EdgeDefl/ of $P$ of magnitude $\delta$
    if and only if $P$ is an \EdgeInfl/ of $Q$ of magnitude $\delta$.
    Let us denote by $P(\Omega) = \{P_x : x\in \Omega\}$.
    Then there exists 
    a space $\overline{\Omega}$ that
    consists of all pairs $(x,s) \in \Omega\times [0,2\delta]$ 
    such that the following conditions hold:
    \begin{itemize}
        \item For any positive \EdgeDefl/ $Q$ of $P_x$, 
        it holds $Q\not\in P(\Omega)$.
        \item The \EdgeInfl/ $(P_x)_s$ of magnitude $s$ is 
        contained in $P(\Omega)$.
    \end{itemize}
    We denote by $P_{x,s}$ the polygon $P$, for which we
    first do the \VertexPert/ $x$ and then inflate by $s$.

\begin{lemma}[Decomposition into Inflations]
\label{lem:Decomposition-Inflation}
    There is a bijection $\varphi$ 
    from $\Omega$ to  $\overline{\Omega}$
    such that $P_x = P_{\varphi(x)}$.
\end{lemma}
\begin{proof}
    For every  \VertexPert/ $P_y$, with $y\in \Omega$ 
    there exists 
    a maximal deflation $P_x$, with $x\in \Omega$. 
    Let $s$ be the amount by which we deflated.
    We define $\varphi(y) = (x,s)$.
    It holds that $P_y = P_{x,s} = P_{\varphi(y)}$
    This describes the bijection $\varphi$.
\end{proof}

Recall that we denote by $B(x)$ the 
required number of bits per guard
after perturbing by~$x$.
By abuse of notation, we denote by $B(x,s)$ the 
corresponding number of bits, for $(x,s) \in \overline{\Omega}$.
Furthermore for each $x \in \Omega$, we denote by~$t_x$, 
the maximum such that~$(x,t_x) \in \overline{\Omega}$.
We denote by $\overline{\Omega}'$, the projection of
$\overline{\Omega}$ onto its first component.
In other words $\overline{\Omega}' \subseteq \Omega $
is the set of all \VertexPert/s such that any
\EdgeDefl/ of them would not yield a polygon in 
$P(\Omega)$.

\begin{lemma}[Integral-Equivalence]
\label{lem:Integral-Equivalence}
    It holds that 
    \[\int_{x\in \Omega} B(x) = 
    \int_{x\in (\overline{\Omega}')} \int_{s\in [0,t_x]} B(x,s).\]
\end{lemma}
\begin{proof}
    By Lemma~\ref{lem:Decomposition-Inflation}, the right and
    the left integral are actually integrating over the same set.
\end{proof}
Using the above bijection, we can now argue that
the number of bits $B(x)$ can be replaced, by the 
expected running time inside the integral,
as we will show in the following lemma.
\begin{lemma}[Use-The-Average]
\label{lem:Averaging}
    It holds that 
    \[ \int_{x\in\Omega} B(x)
    \leq
    \int_{x\in\Omega}   O\left(  \log(Ln/t_x) \right).\]
\end{lemma}
\begin{proof}
    We start by expanding the integral using Lemma~\ref{lem:Integral-Equivalence}.
    \begin{align*}
        \int_{x\in\Omega} B(x) 
            & = \int_{x\in \overline{\Omega}'} \int_{s\in [0,t_x]} B(x,s)\\
            &\leq_c \int_{x\in \overline{\Omega}'} t_x \,   \log(Ln/t_x)
    \end{align*}
    The last line follows from Lemma~\ref{lem:ExpectedBitInflations},
    as the inner integral corresponds to an
    \EdgeInfl/.
    Using  that $  t c =\int_{s\in[0,t]} c$, for every constant $c$, we get
    \begin{align*}
            &= \int_{x\in \overline{\Omega}'} \int_{s\in [0,t_x]}  \log(L n /t_x) \\
            &= \int_{x\in \overline{\Omega}'} \int_{s\in [0,t_x]}  \log(L n /(t_x - s) 
    \end{align*}
        Note that for $(x,s) = \varphi(x')$, it holds that
        $t_x -s = t_{x'}$. 
        (We refer to $\varphi$ from 
        Lemma~\ref{lem:Decomposition-Inflation}.)
        This implies
    \begin{align*}
            &= \int_{x\in\Omega} \log(L n/t_x).
    \end{align*}
    Here the last line 
    follows from the bijection 
    explained before Lemma~\ref{lem:Integral-Equivalence}.
\end{proof}

\begin{lemma}\label{lem:Expected-Vertex-Perturbation-Time}
    Let $P$ be 
    a polygon
    and $\delta$ denotes the magnitude
    of a \VertexPert/.
    Furthermore, assume that the polygon $P$ has $n$ 
    vertices and fits inside a square $[0,L]^2$ and has pointedness $\beta$.
    Then the 
    expected number of bits is upper bounded by
    \[O\left(  \log \frac{n L}{\delta \beta}     \right).\]
\end{lemma}
\begin{figure}[tbph]
        \centering
        \includegraphics{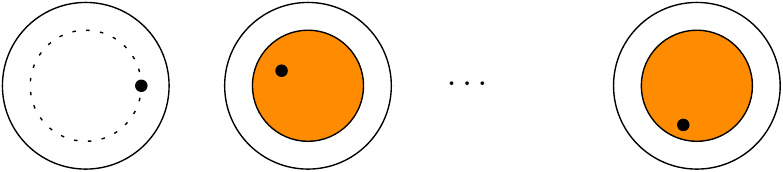}
        \caption{Here, one vertex has distance $t$ and all other vertices
        have distance at most $t$ from its original vertex.}
        \label{fig:Vertex-Integral}
    \end{figure}
\begin{proof}
    Let us start with the definition of the 
    expected running time and Lemma~\ref{lem:Averaging}.
\begin{align*}
    \E (B) &= \frac{1}{(\pi \delta)^n} \int_{x\in \Omega} B(x) \\
            &\leq_c \frac{1}{(\pi \delta)^n} \int_{x\in \Omega}  \log ( Ln /t_x ).
\end{align*}

    Note that $\Omega$ is
    the Cartisian product of $n$ disks.
    And we think of $x$ as picking $n$ points from
    $n$ disks independently.
    See Figure~\ref{fig:Vertex-Integral} for an illustration.
    Similar to the proof of Lemma~\ref{lem:Expected-Edge-Perturbation-Time}, 
    we can think of a new random experiement to
    select $x = (x_1,\ldots,x_n)\in\Omega = \disk(\delta)^n$ as follows.
    First guess a number $t\in[0,\delta]$,
    which represents the maximum distance of
    every point $x_i$ to its disk center, and guess, 
    which of the $n$ points is attaining this maximum
    distance. The remaining points have all at
    most distance $t$ from the center.
    
    Thus we can rewrite the integral as follows.
\begin{align*}
    \int_{x\in \Omega}  \log ( Ln/t_x )
    &= n \int_{t\in [0,\delta]} 2\pi t 
    \int_{x'\in \,{\disk(t)}^{n-1}} \log ( Ln/t_x )\\
\end{align*}   
Here the factor $2\pi t$ comes from the circumference of
a disk of radius $t$.
Using that $\beta (\delta - t) \leq t_x$, 
by Lemma~\ref{lem:pointedness}, we get
\begin{align*}   
    &\leq_c n \int_{t\in [0,\delta]} 2\pi t 
    \int_{x'\in \, {\disk(t)}^{n-1}} \log\frac{Ln}{\beta (\delta - t)}
\end{align*}   
Note that the inner integral does not 
depend on $x'$. Furthermore the area 
of a disk of radius $t$ equals $\pi t^2$, which gives us
\begin{align*}
    &= n \int_{t\in [0,\delta]} 2\pi t 
    (\pi t^2)^{n-1} \log\frac{Ln}{\beta (\delta - t)} \\
    &= 2 n \pi^n \int_{t\in [0,\delta]} t^{2n-1} \log\frac{Ln}{\beta (\delta - t)}\\
\end{align*}
    Using some computer algebra system, we can compute the integral.
    This is exactly the same computation as in Lemma~\ref{lem:Expected-Edge-Perturbation-Time}, and gives us
\begin{align*}   
    &\leq_c 2n   \, \pi^n  \frac{\delta^n\, \log \frac{ L n^2}{\beta \delta } }{2n}\\
    &=  \pi^n \delta^n\, \log \frac{ L n^2}{\beta \delta }. 
\end{align*}

This gives
    \[
    \E(B)    \leq_c \frac{1}{(\pi \delta)^n}      \, \pi^n \, \delta^n \,\log \frac{n L}{\beta \delta }
    =  O\left( \log \frac{n L}{\beta \delta }\right).
    \qedhere 
    \]
\end{proof}

\section{Expected Non-deterministic Time}
\label{sec:ExpectedNDT}

This section is devoted to the proof of the following theorem.
\Upper*

Note that the non-deterministic part is easy. There are at most
$n$ guards in an optimal guarding and each can
be described with a logarithmic number of bits, according
to Theorem~\ref{thm:MainBits}.
One may wonder whether it is possible to improve
the running time to the form $O(k\log \ldots)$,
where $k$ is the optimal number of guards.
The problem is that $k$ is not well-defined, as it 
may vary depending on the perturbation.

It remains to describe a deterministic algorithm to check
if a given set of guards is indeed guarding a given polygon.
It was shown by Efrat and Har-Peled~\cite{EfratH06}
that for a {\em simple} polygon $P$ we can check in
$O(kn\log k \log n)$ time if $k$ given guards 
see $P$ completely.
\begin{lemma}[\cite{EfratH06}]
    It can be checked in $O(kn\log k \log n)$ time
    if a given set of $G$ guards is correctly
    guarding a given polygon, without holes, on $n$ vertices.
    This algorithm works on a \realRAM/.
\end{lemma}

For polygons with holes, we describe a simple 
(probably well-known) algorithm that runs in $O(n^2k^2)$ time.
Here $k$ is defined as the number of given guards.
As we have not found a reference in the literature, we
repeat it here for the benefit of the reader.
\begin{lemma}[Folklore]\label{lem:CheckGuard4}
    It can be checked in $O(k^2n^2)$ time
    if a given set of $G$ guards is correctly
    guarding a given polygon, potentially with holes, on $n$ vertices.
    This algorithm works on a \realRAM/.
\end{lemma}
\begin{proof}
 For polygons with $h$ holes, it is possible to compute the 
visibility regions of each guard in $O(n+h\log h)$ 
time~\cite{VisibilityHoles}. Note
that the number of vertices and edges of those polygons is still $O(n)$.
(Every pair of adjacent edges of the visibility 
polygon contains at least one vertex
of the original polygon and every vertex is incident to at most two
edges.)
Then we have in total $m = O(kn)$ edges and vertices.

We can now compute the union $Q$ of all those polygons in
$O(m\log m + l)$ time, where $l = O(m^2) = O(k^2 n^2)$ is the total number of 
edge intersections of all the given edges.
This can be done by a simple sweepline algorithm of all the
visibility polygons~\cite[Chapter 7]{o1998computational}.

Thereafter, we can check if the set of vertices of $Q$ are the same
(in the same order) as the vertices of our original polygon $P$.
This can be done in linear time.
The running time is dominated by $O(k^2n^2)$, by taking the union
of the visibility polygons.   
\end{proof}

\section{Non-deterministic Time with High Probability}
\label{sec:With-High-Probability}

\UpperProbability*

    We first focus on the non-deterministic part of the algorithm.
    Let us assume that there are some numbers 
    $0 = t_0< t_1<\ldots < t_\ell = \delta$
    such that for all $i$ and $s\in [t_{i-1},t_{i})$
    it holds that $|OPT(P_s)|$ is constant.
    As $| OPT(P_s)| $ is 
    monotonically decreasing it holds that $\ell \leq n$.
    We denote by $\delta_i = t_i - t_{i-1}$.
    Consider the set 
    \[S 
    = \bigcup_{i=1,\ldots,\ell} \, [t_{i-1}+\tfrac{\delta p}{2n} , t_i] .\]
    Note that $\lvert S \rvert \geq \delta(1 - p/2)$.
    Let us define $\overline{S} = [0,\delta] \setminus S$.
    Note that $|\overline{S} \cap \Omega_{\delta,q}|\leq n + p(q+1)/2 $.
    Thus, for $q>2n/p$, it holds that 
    $|\overline{S} \cap \Omega_{\delta,q}|\leq p(q+1) = p\, |\Omega_{\delta,q}|$.
    Therefore a random perturbation will be in the set
    $S$ with probability at least $1-p$.
    Now, fix some $s\in S$ and $i \in \{1,\ldots,\ell\}$,
    with $s\in [t_{i-1}+\delta p/2n, t_i)$.
    Note that $\lvert OPT(P_{t_{i-1}})\rvert  = \lvert OPT(P_s)\rvert$.
    Furthermore $P_s$ is an \EdgeInfl/ of $P_{t_{i-1}}$ of magnitude
    at least $\delta p /2n$. Thus, by Lemma~\ref{lem:Fixed-Edge-Inflation}, 
    a grid of width $w = \sqrt{2}\delta p /2n = \delta p /\sqrt{2}n$ contains an optimal
    guarding of $P_s$. We can describe each  
    guard with 
    \[\lceil \log L\rceil + \lceil\log (\sqrt{2}n / \delta p)  \rceil  
    = O\left(\log \frac{L n}{\delta p}\right)\]
    bits. Thus clearly $O(n\log \frac{L n}{\delta p})$ bits in total
    are sufficient as any polygon on $n$ vertices 
    can be guarded by at most $n$ guards.
    This finishes the non-deterministic part of the algorithm.
    
    It remains to argue that we can also 
    check optimal guards on a \TuringMachine/.
    Recall that we mentioned in Lemma~\ref{lem:CheckGuard4}
    how to check a given set of guards in
    $O(n^4)$ time on a \realRAM/.
    It remains to argue that this algorithm 
    also runs in polynomial time on a \TuringMachine/.
    \begin{figure}[tbph]
        \centering
        \includegraphics{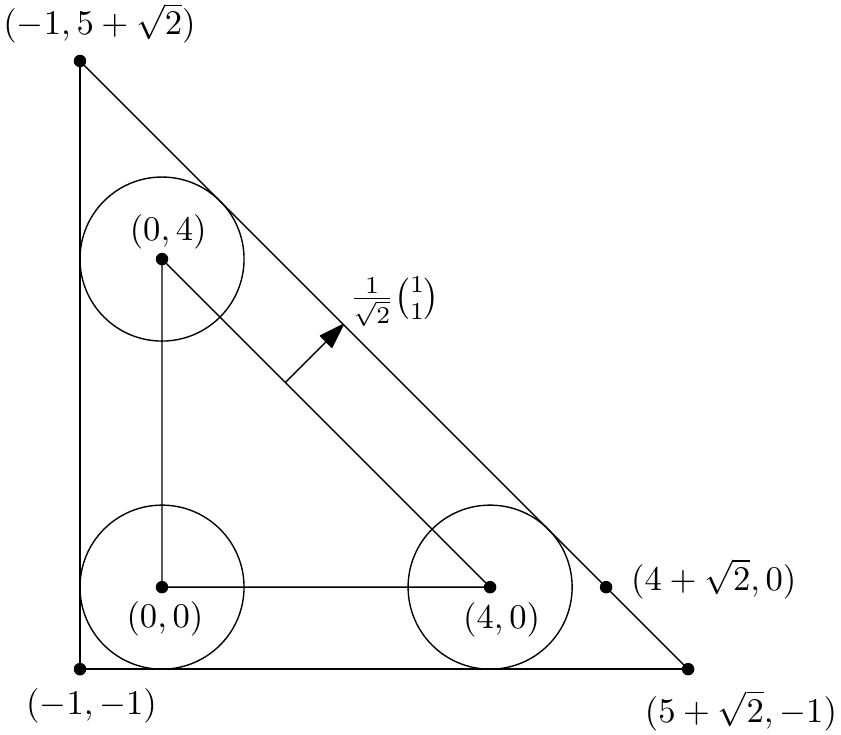}
        \caption{A triangle with integer vertices and an \EdgeInfl/ by $1$.
        The \EdgeInfl/ has vertices that can be
        described using radicals.}
        \label{fig:Triangle-Inflation}
    \end{figure}
    First note that even after a rational perturbation,
    the resulting vertices may not be 
    rational, see Figure~\ref{fig:Triangle-Inflation}.
    Let us start with the way that a vertex can be represented.
    For that purpose, suppose that $r\in \Omega_{\delta,q}$ 
    describes an \EdgeInfl/. Assume that $u,v,w \in \Q^2$ are three
    consecutive vertices of the original polygon $P$.
    We will denote the new vertices by $u',v',w'$.
    We can compute the normal vectors of the edges $uv$
    and $vw$ using square roots. 
    Using the normal vectors, we can compute the slope
    and the $y$-intercept of the supporting lines
    of $u'v'$ and $v'w'$. From there, we can compute the intersection
    point of those lines, which define the point $v'$.
    The resulting point can be described using square roots.
    To be more precise, the only place we will leave the field
    $\Q$ is, when we compute the normal vectors. 
    Denote by $a$ and $b$ the Euclidean distance
    between $u,v$ and $v,w$ respectively.
    Then all calculations take place in the field $\Q[a,b]$.
    
    Now, consider the algorithm described in the 
    proof of Lemma~\ref{lem:CheckGuard4}.
    The first set $A$ of geometric objects that is computed are
    segments defined by combinations of guards and 
    vertices.
    The second set $B$ of objects that is computed are 
    intersections of segments in $A$.
    Note that each object in $A\cup B$ can be 
    defined using a constant number of vertices and/or
    guards. No further geometric objects are computed.
    Once all those geometric objects are computed,
    all further queries on the coordinates of 
    those objects are just comparisons of $x$ and $y$-coordinates.
    Such a comparison reduces to solving a 
    {\em sum of square roots problem}, which is not in general
    known to be polynomial time solvable~\cite{OpenProblemProject33}.
    It is one of the major open problems in 
    computational geometry, to find a polynomial
    time algorithm.
    In our case however, we know that only a finite 
    number of square roots are involved, and this case 
    is known to be polynomial time solvable~\cite{qian2006much}.

    Note that the time of the deterministic part of the algorithm 
    depends polynomially on $\log q$.

\section{Non-deterministic Time Using Algebra}
\label{sec:NPtime}
This section is devoted to describe an
algorithm that runs in expected deterministic time.
The advantage of the next theorem is that it really
works on a non-deterministic \TuringMachine/, without
a \realRAM/. And the described algorithm works correctly 
in all cases, not just with high probability. 
However, the algorithm uses algebraic methods in a very small
fraction of the cases.

\NPTime*
\begin{proof}
    We use the \NaiveAlgo/ as described in 
    Theorem~\ref{thm:MainHighProb} with probability 
    $1- p = 1- n^{-cn}$,
    where $c$ is some sufficiently large constant
    to be determined later.
    The non-deterministic running time in that case equals
    $O(n\log\frac{Ln}{\delta p }) = O(n^2\log\frac{Ln}{\delta})$.
    In the remaining cases, we are using algebraic algorithms
    that run in $n^{c'n}$ time~\cite{basu2006algorithms,EfratH06,EfratH02}.
    We set $c = c'$.
    The expected non-deterministic running time can be estimated by
    \[(1-p)\  O\left(n^2\log\frac{Ln}{\delta}\right) + pn^{cn} =  O(n^2\log\frac{Ln}{\delta}).\]
    The deterministic time to check if the guess made
    by the \NaiveAlgo/ is correct is the same as 
    in Theorem~\ref{thm:MainHighProb}.~\qedhere
\end{proof}

\section{Discussion of Perturbation Models}
\label{sec:DiscussPertubations}
In this section we compare various models of perturbation.
This is important to critically evaluate
the contribution of our findings. See Table~\ref{tab:Perturbation-Quality},
for a summary.

\begin{table}[htbp]
    \centering
    \begin{tabular}{l|cccc}
         & Vertex-Pertur. & Minkowski-Infl. & Edge-Infl. & Edge-Pert. \\
         \hline
        remains a polygon & YES & NO & YES & YES \\
        number of vertices preserved & YES & NO & YES & YES \\
        angles preserved  & NO  & YES & YES & YES \\
        shape preservation & YES  & YES & YES & YES \\
        dimension of randomness & $2n$ & $1$ & $1$ & $n$ \\
        small Hausdorff distance & YES & YES & NO & NO \\
        fair to edges & NO & YES & YES & YES \\
        fair to vertices & YES & YES & NO & NO \\
        small visibility change & YES & NO & NO & NO 
    \end{tabular}
    \caption{A comparison of the four models of perturbation.}
    \label{tab:Perturbation-Quality}
\end{table}

At first glance, \VertexPert/s seem to be the
most natural type of perturbations, as we usually 
specify a polygon by describing its vertices.
However, bear in mind that practical instances are constructed
with specific features in mind. 
Also recall that \SmoothAna/ aims to analyze
instances as similar as possible to the input.
Thus, if a perturbation destroys a constructed feature,
it is less meaningful then a perturbation that maintains that feature.
In fact, if we would not alter the instance at all,
then we would simply get the worst case analysis.
Let us now have a second glance at the various
ways to perturb a polygon, and see which behave most nicely
under the described criteria.

One of the most natural criteria is whether we maintain 
the property of being  
a polygon after the perturbation. Clearly, \MinkInfl/s
violate this criteria. However, note that we could also take the
Minkowski-sum of a polygon with a square, and the resulting object would remain polygon.

We might also like to know whether the number of vertices remains. Again, this fails only for the \MinkInfl/s.

A feature that seems to be important, is whether the 
angles between edges are maintained.
It is, for instances, often the case in buildings that
walls are rectilinear to one another. We might not care 
about the precise width of a corridor, but we might find it odd
if the two walls of the corridor are not parallel. 
Here, only the \VertexPert/ fails miserably.

Another criteria is whether the general shape
is maintained. 
All presented perturbations satisfy this criteria.
But one could think of perturbations where
the shape is not maintained at all.
For instance by permuting the order of the vertices.

As already mentioned, we want to perturb as little 
as possible. We can take the dimension of the probability
space as measure for the added randomness.
Here, the inflations are best, as their probability
space is one dimensional.
However, note that one could argue that the perturbation
is designed to destroy irrational solutions.

Another important aspect is by how much visibility is altered
by the perturbation. One can see that the visibility regions
are altered more by \MinkInfl/ and \EdgePert/,
in comparison to the \VertexPert/.
The reason is that reflex vertices might be moved a lot
in case that the inner angle at that vertex is particularly large.
This might be the reason, why our analysis is easier
    with the \MinkInfl/. See Figure~\ref{fig:Visibility-Altered} for an illustration.
\begin{figure}[htbp]
    \centering
    \includegraphics{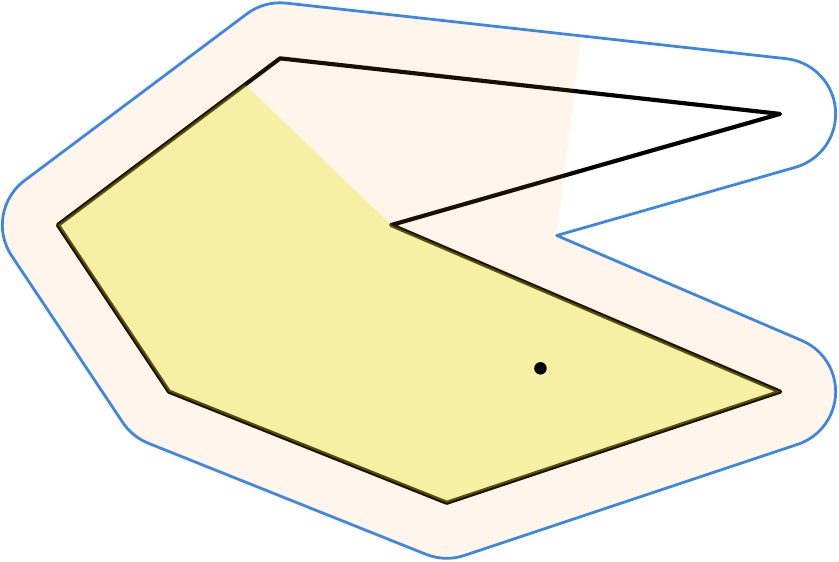}
    \caption{Left: The visibility region after some \VertexPert/ is 
    altered comparatively little. Right: The \MinkInfl/ alters the
    visibility a lot. This might be the reason, why our analysis is easier
    with the \MinkInfl/.}
    \label{fig:Visibility-Altered}
\end{figure}

A somewhat weird criteria, is a voting of the points on the boundary.
After a perturbation of magnitude $t$, every point could be 
asked the following question: 
``Do you feel perturbed by magnitude $t$?''
For the \EdgeInfl/, it is easy to see that most points
on the edge, will say yes, whereas vertices and points close
to vertices might say, that they feel perturbed much more then 
$t$. On the other hand, in the \VertexPert/ model all vertices,
will say that they are moved by $t$, whereas most points
on the edges might say that they are moved far less.
Thus loosely speaking, \VertexPert/s treat edges
in a ``dishonest manner'', whereas \EdgePert/s treat
vertices in a ``dishonest manner''.
It seems difficult to design a model of perturbation
that is both fair to the edges and to the vertices.
(Combining Edge and \VertexPert/, is unfair to 
both edges and vertices.)

Finally, one might be interested in maintaining a small
Hausdorff-distance to the original polygon after the
perturbation. This is not the case for \EdgeInfl/ and 
\EdgePert/s as can be seen in 
Figure~\ref{fig:Big-Hausdorff}.

\begin{figure}[htbp]
    \centering
    \includegraphics{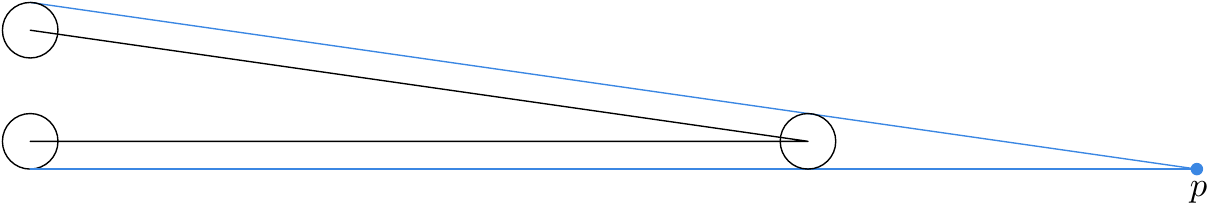}
    \caption{Altough, the \EdgeInfl/ is fairly small the 
    Hausdorff distance to the perturbed polygon is 
    comparatively large, due to the very small convex angle.}
    \label{fig:Big-Hausdorff}
\end{figure}

There is another aspect that we so far have swept totally
under the rug. Spielman and Teng did not considered 
uniform distributions, but Gaussian distributions.
Their argument being that practical instances are often 
altered by random {\em noise} and noise is for various
reasons best modeled by the Gaussian distribution.
We believe that the Gaussian distribution makes life 
considerably more difficult and we don't expect any
difference in the results. 
The reason is that both Gaussian and uniform distribution
choose from the same set of instances. In both cases,
the result is that under both distributions the majority of instances
behaves good. The only difference is that
we ``drew'' instances with some small difference in distribution.

We summarize that none of the considered
models of perturbation is ideal in all criteria
and preference of one over the other might be
a question of taste.
The authors find the preservation of angles and
low added randomness most important.

\section{Open Problems and Conclusion}
\label{sec:Conclusion}

In this paper, we have shown that {\em typically}
a polygon can be guarded optimally by guards
with rational coordinates and those coordinates
might be very small.
This is one of the few positive
results on the \AGP/. 

As a corollary, the \NaiveAlgo/ takes linear non-deterministic time
and $O(n^4)$ deterministic time on polygon with holes.
Thus, we ask:
\begin{question}\label{quest:Deterministic}
    Does there exists an $O(kn)$ algorithm, to decide if
    a given set of $k$ guards, is correctly guarding
    a given polygon with holes on $n$ vertices?
\end{question}

While, we showed that \SmoothAna/ can overcome
the obstacle of \ER/-hardness for the \AGP/,
most lower bounds, like \NP/-hardness, W[1]-hardness,
and inapproximability are expected to continue to hold, even 
after perturbations. 
Note that practical implementations usually consist of 
a geometric part, which is theoretically 
polynomial time solvable and a combinatorial
part, which solves a set-cover problem, which is NP-hard in general.
De Rezende et al.~\cite{engineering} report as follows:
\begin{adjustwidth}{1cm}{1cm}
 ``Still the geometric routines made up for over
90\% of the runtime.''
\end{adjustwidth}
In other words their algorithms spend most time on the geometric part
and least time on the part of the algorithm that is \NP/-hard.
It would be interesting to find a theoretical explanation
for this phenomenon.
\begin{question}
    Does there exist a theoretical explanation for the 
    reported findings, that the \NP/-hard part of the implemented
    algorithms takes the least time?
\end{question}


We showed that the grid is a candidate set in
the \SmoothAna/ model. 
While it is nice to have a candidate set at all,
we would really like to have one of polynomial size.
All the \NP/-hardness proofs that the authors are aware of
yield instances with easily identifiable polynomial 
sized candidate sets. Thus it is not clear, if such a candidate
set might not exist for every polygon, after a small perturbation.
\begin{question}
    Does there exist an algorithm that runs in smoothed
    polynomial time and outputs a candidate set
    of polynomial size?
\end{question}

In this paper, we studied the so-called \NaiveAlgo/,
which is the simplest possible algorithm for the 
\AGP/. But there are implementations of much more clever
algorithms that perform well in practice.
\begin{question}
    Can we show that one of the more sophisticated practical 
    algorithms works correctly in the smoothed sense?
\end{question}

Here, we studied the \AGP/, but there are many other
\ER/-hard problems. Maybe, we can find that they
also have a non-deterministic polynomial time
algorithm, which runs correctly in the 
\SmoothAna/ sense.
A possible candidate is Motion Planning.
Note that there are many variants of Motion
Planning, and we currently do not even know
which variants are \ER/-hard.
But often algebraic methods are the only way to solve
the problem provably correctly~\cite{basu2006algorithms}.
Another example would be Nash-Equilibria~\cite{SchaeferS17}.

As we have shown that essentially all 
models of perturbation behave in the same way, we think 
that it should be sufficient in the future to regard
\EdgeInfl/ only. This model is very nice as it preserves
angles, the probability space is one dimensional
and it is usually easiest to analyze. 
Furthermore, our techniques strongly suggest that 
any result on \EdgeInfl/ carries over to the 
other models, in a tedious but standardized way.

    \paragraph{Acknowledgments}
    We want to thank Stefan Langerman, Jean Cardinal, John Iacono 
    and Mikkel Abrahamsen for helpful discussion on the presentation
    of the results. We want to thank Joseph O'Rourke for 
    pointing us to an algorithm to check if a given
    set of guard positions is correct.

Part of the research was conducted while visiting KAIST, and we thank KAIST and BK21 for their support and for providing an excellent working environment.

Andreas Holmsen was supported by the Basic Science research Program through the National Research Foundation of Korea (NRF) funded by the Ministry of Education (NRF-2016R1D1A1B03930998).

    Tillmann Miltzow acknowledges the generous support from the ERC Consolidator Grant 615640-ForEFront
    and the Veni EAGER.
    
\bibliographystyle{abbrv}
\bibliography{references}
\end{document}